\documentclass[12pt]{article}
\usepackage[top=20truemm,bottom=20truemm,left=30truemm,right=30truemm]{geometry}
\usepackage{enumerate,indentfirst,braket,mathrsfs,mathtools,amsmath,amssymb,comment,bm,amsthm,amsfonts}
\usepackage{graphicx, color}
\usepackage{float}
\usepackage{wrapfig}
\usepackage{ascmac, fancybox}
\usepackage{bm}
\usepackage{bbm} %定義関数を表示
\usepackage[toc, page]{appendix}

\usepackage{authblk} %複数著者用

\usepackage{multirow}
\usepackage{diagbox}
\usepackage{tikz}
\usetikzlibrary{positioning}
\usetikzlibrary{calc}

\theoremstyle{plain}
\newtheorem{thm}{Theorem}[section]
\newtheorem*{thm*}{Theorem}

\theoremstyle{definition}

\theoremstyle{remark}

\newtheorem*{con*}{Conjecture}

\usepackage{ulem}

\title{A game-theoretic probability approach to loopholes in CHSH experiments}
\author[1]{Takara Nomura}
\author[2]{Koichi Yamagata}
\author[3]{Akio Fujiwara}
\affil[1]{Department of Mathematics, The University of Osaka, Toyonaka, Osaka, 560-0043, Japan (u004883d@ecs.osaka-u.ac.jp)}
\affil[2]{Institute of Science and Engineering, Kanazawa University, Kanazawa, Ishikawa, 920-1192, Japan (yamagata@se.kanazawa-u.ac.jp)}
\affil[3]{Department of Mathematics, The University of Osaka, Toyonaka, Osaka, 560-0043, Japan (fujiwara@math.sci.osaka-u.ac.jp)}

\date{}

\begin{document} 
\maketitle

\begin{abstract}
We study the CHSH inequality from an informational, timing-sensitive viewpoint using game-theoretic probability, which avoids assuming an underlying probability space. 
The locality loophole and the measurement-dependence (``freedom-of-choice'') loophole are reformulated as structural constraints in a sequential hidden-variable game between Scientists and Nature.
We construct a loopholes-closed game with capital processes that test (i) convergence of empirical conditional frequencies to the CHSH correlations and (ii) the absence of systematic correlations between measurement settings and Nature's hidden-variable assignments, and prove that Nature cannot satisfy both simultaneously: at least one capital process must diverge.
This yields an operational winning strategy for Scientists and a game-theoretic probabilistic interpretation of experimentally observed CHSH violations.

\end{abstract}

%==================================================================
\section{Introduction}\label{sec:Introduction}
%==================================================================

Bell's inequality was formulated in the context of the EPR paradox \cite{EPR_paradox}, and it characterizes correlations that are compatible with local hidden-variable models (see, e.g. Bell's seminal work \cite{Bell_EPR} and comprehensive reviews \cite{Brunner2014}).
Among its variants, the CHSH inequality \cite{CHSH} has played a central role.
It has been extensively studied in foundations of quantum mechanics and subjected to thorough experimental tests \cite{Aspect1981, Aspect1982, Weihs}.

In certain experimental configurations, however, it cannot be ruled out that a system may appear to violate the CHSH inequality even though it follows a hidden-variable model.
Such possibilities are known as ``loopholes'' and the identification and the control of loopholes in CHSH-type experiments have been a topic of sustained theoretical and experimental interest \cite{Bell_theorem_encyclo,Kaiser_trackling_loopholes}.

This paper focuses on two representative loopholes, the locality loophole and the measurement-dependence loophole, also known as the freedom-of-choice loophole.
The locality loophole refers to the possibility that the outcomes at spatially separated measurement sites may influence one another.
This loophole is addressed experimentally by performing measurements at sufficiently distant locations, where the measurement results do not exert an immediate influence.
The measurement-dependence loophole refers to the possibility that the choice of measurement settings is statistically correlated with a hidden variable.
Experimentally, this loophole is considered closed by selecting the measurement settings ``freely'' or ``randomly,'' for example using physical random number generators after particle generation \cite{Violation_local_realism_freedom_of_choice}.
Modern ``loophole-free'' Bell tests that simultaneously aim to close these loopholes have reported violations of the CHSH inequality consistent with quantum-mechanical predictions \cite{Hensen, Giustina}.

While most analyses of the CHSH inequality rely on probabilistic or measure-theoretic formulations, comparatively little attention has been paid to how the availability and timing of information themselves constrain possible correlations.
The objective of this paper is to shed new light on the conditions under which the CHSH inequality holds or is violated, from this informational perspective.
To this end, this paper employs the framework of game-theoretic probability theory \cite{itsonlyagame,Shafer-Vovk_Foundations}.

Game-theoretic probability theory models statistical behavior as arising from a betting game, without introducing an underlying probability space.
Within this framework, players act sequentially, so that one player may choose a move after observing the other's.
This asymmetry in information and timing plays a central role in the emergence of statistical behavior.
This explicit representation of these structural features enables us to reconsider the assumptions underlying the CHSH inequality from an alternative perspective.

In this paper, we formalize the locality loophole and the measurement-dependence loophole as structural conditions in games between Scientists and Nature inspired by hidden-variable scenarios.
We then introduce a loopholes-closed game equipped with two capital processes.
One is designed to test whether empirical conditional frequencies of outcomes are consistent with the CHSH correlations.
The other probes whether, across trials, the chosen measurement settings exhibit any systematic dependence on Nature's outcome assignments.
Importantly, this second test is formulated in terms of empirical frequencies, providing an operational non-probabilistic analogue of the usual measurement-independence assumption.

Our main result shows that these two requirements cannot be satisfied simultaneously within this game-theoretic framework.
This incompatibility leads to a formulation in which violations of the CHSH inequality can be understood operationally, in terms of sequential betting strategies, once the relevant structural conditions are enforced.

This paper is organized as follows.
Section 2 introduces a fundamental result from game-theoretic probability theory.
Section 3 is devoted to preliminaries for translating the CHSH setting into the framework of game-theoretic probability theory.
Section 4 presents the main results, clarifying how the loopholes are incorporated into the game-theoretic framework and how these structural conditions affect the validity or violation of the CHSH inequality.
For improved readability, the detailed proofs are deferred to Appendix A.
Section 5 provides concluding remarks.
Appendix~B provides a brief summary of the CHSH inequality for convenience.

%==================================================================
\section{Game-theoretic probability theory: a minimal background}
%==================================================================

Let $\Omega := \{1, 2, \dots,A\}$ be a finite alphabet, and let $\Omega^n$ and $\Omega^\ast$ denote the sets of sequences over $\Omega$ of length $n$ and finite length.
An element of $\Omega^n$ is represented symbolically as $\omega^n$.
Let us introduce the following set: 
\begin{displaymath}
    \mathcal{P}(\Omega):=\left\{ p:\Omega\to (0,1) \left| \; \sum_{\omega\in\Omega} p(\omega)=1 \right.\right\}.
\end{displaymath}
Given a $p\in\mathcal{P}(\Omega)$, consider the following game.

\medskip
\begin{itembox}[l]{\bf Simple predictive game} \label{game:predictive_by_Shafer_and_Vovk}
\textbf{Players} : Skeptic and Reality. \\
\textbf{Protocol} : $K_0 = 1$. \\
\hspace*{0.5em}FOR $n \in \mathbb{Z}_{>0}$ : \\
\hspace*{1em}Skeptic announces $q_n \in \mathcal{P}(\Omega)$. \\
\hspace*{1em}Reality announces $\omega_n \in \Omega$. \\
\hspace*{1em}$\displaystyle K_n := K_{n-1} \cdot \frac{q_n(\omega_n)}{p(\omega_n)} $. \\
\hspace*{0.5em}END FOR.
\end{itembox}
\medskip

This protocol can be understood as a betting game in which Skeptic predicts Reality's ``stochastic'' move, regarding $p(a)$ as the ``probability'' of the occurrence of $a\in\Omega$.
These quantities are multiplicative update factors for Skeptic's capital, and we therefore regard them as the betting odds.
Further, $q_n$ and $K_n$ denote Skeptic's bet and capital at step $n$, respectively, with the recursion formula specifying how the capital evolves.
Note that, for any $n$, $K_n > 0$ from the recursion formula of Skeptic's capital.
Since $q_n$ can depend on Reality's past move $\omega_1^{n-1} = \omega_1 \omega_2 \cdots \omega_{n-1} \in \Omega^{n-1}$, we identify Skeptic's strategy $\{q_n\}_n$ with a map $q: \Omega^\ast \rightarrow \mathcal{P}(\Omega)$ as $q_n(a) := (q(\omega_1^{n-1})) (a).$

Apparently, this game is in favor of Reality because Reality announces $\omega_n$ {\it after} knowing Skeptic's bet $q_n$, preventing Skeptic from becoming rich.
However, Shafer and Vovk \cite{itsonlyagame} showed the following surprising result. 

\begin{thm} [Game-theoretic law of large numbers] \label{SLLN:predictive-Shafer_and_Vovk}
In the simple predictive game, Skeptic has a strategy $q : \Omega^\ast \rightarrow \mathcal{P}(\Omega)$ that ensures $\lim_{n \to \infty} K_n = \infty$ unless
\begin{displaymath}
    \lim_{n \to \infty} \frac{1}{n}\sum_{i=1}^n \delta_a (\omega_i) = p(a)
\end{displaymath}
for all $a \in \Omega$, where $\delta_a$ denotes the Kronecker delta.
\end{thm}

For a proof, we refer the reader to the original textbook \cite{itsonlyagame,Shafer-Vovk_Foundations}. 
A more information-theoretic perspective is discussed in \cite{coding_study}.

The theorem implies that there exists a betting strategy $q_n$ that guarantees Skeptic becomes infinitely rich if Reality's moves violate the ``law of large numbers.''
\footnote{If Reality's objective is to prevent Skeptic from becoming infinitely rich, then Theorem~\ref{SLLN:predictive-Shafer_and_Vovk} can be rephrased as follows: {\it Skeptic can force the event $\lim_{n \to \infty} \frac{1}{n}\sum_{i=1}^n \delta_a (\omega_i) = p(a)$}, which corresponds to the original statement by Shafer and Vovk \cite{itsonlyagame}.
}
Note that Skeptic's capital may still diverge even while the empirical frequencies of Reality's moves converge; the two phenomena are not mutually exclusive \cite{coding_study}.

In protocols appearing in game-theoretic probability theory, it is assumed that agents declaring later can utilize the moves of players who declared earlier.
This timing convention will be adopted throughout the paper.
In the following sections, we reinterpret Skeptic as Scientist and Reality as Nature, in order to emphasize the correspondence with experimental settings.

%==================================================================
\section{From CHSH scenario to a game-theoretic framework}
%==================================================================

In this section, we explain the physical and game-theoretic situations we study.
Scientist A and Scientist B are located at spatially separated sites.
Each receives a particle, chooses one of two measurement settings, and performs the corresponding 
measurement; Nature then produces an outcome in $\Omega := \{+1, -1\}$ for each party.

In preparation for a game-theoretic interpretation of the CHSH inequality, we introduce a hidden-variable-style representation in which outcomes for all four measurement settings are preassigned.
Let $\Lambda$ be a set, and for each $s \in \{1, 2, 3, 4 \}$, let $X_{s}:  \Lambda \rightarrow \Omega$ be a function.
We assume that Nature fixes the functions $X_1,\dots, X_4$ prior to the game protocol and that they are disclosed to all participants before the game begins.
This differs from the usual hidden-variable narrative, where the hidden variable $\lambda \in \Lambda$ is not accessible to the experimenters;  
however, this difference is inessential for the arguments below, since we will not exploit any informational advantage from knowing $\lambda$ itself.

In the standard CHSH formulation, $\Lambda$ is equipped with a probability measure
and each $X_s$ is a random variable.
Here, by contrast, we impose no probabilistic structure: $\Lambda$ is merely a set and each $X_s$ is simply a function.
Our only structural assumption is that the map
\[ \lambda \mapsto (X_1(\lambda), X_2(\lambda), X_3(\lambda), X_4(\lambda)) \]
is surjective onto $\Omega^4$; equivalently, for every $(\omega^1, \omega^2,\omega^3, \omega^4 ) \in \Omega^4$, there exists $\lambda \in \Lambda$ such that $X_s (\lambda) = \omega^s$
for all $s\in\{1,2,3,4\}$.

Let
\[
\Theta := \Theta^A \times \Theta^B := \left\{ 1, 2 \right\} \times \left\{ 3, 4 \right\}.
\]
Here, $\Theta^A=\{1,2\}$ and $\Theta^B=\{3,4\}$ denote the sets of measurement settings available to Scientist A and Scientist B, respectively.
The product $\Theta=\Theta^A\times\Theta^B$ therefore represents the set of all joint measurement settings $(\theta^A, \theta^B)$ that can occur in each measurement.

We model Scientists and Nature as players in a game-theoretic setting, who place bets on the outcomes produced by Nature under these measurement choices.
In this setting, the betting odds depends on the measurement settings and the outcomes generated by Nature.
Table~\ref{table:CHSH_odds} specifies the betting odds assigned to each experimental outcome for each pair of measurement settings; we call it the {\it CHSH table}.
The parameters \(\mu := (2- \sqrt{2})/8\) and \(\nu := (2+ \sqrt{2})/8\) are chosen so as to match the quantum-mechanical prediction for the standard CHSH configuration achieving the maximal violation (the Tsirelson bound; see~\cite{Tsirelson}).

For example, suppose Scientists A and B choose the measurement setting pair \((\theta^A,\theta^B)=(1,4)\). 
Then the odds assigned to the outcome \((\omega^A,\omega^B)=(+1,-1)\) are \(\mu\); we write \(p(+1,-1\mid 1,4)=\mu\).

Although the CHSH table may be read as a table of ``conditional probabilities,'' this probabilistic interpretation plays no role in our game-theoretic framework.
What matters for the game is only that Scientists agree on these odds before betting.

\begin{table}
\begin{center}
\begin{tabular}{|l|c|c|c|c|}
    \hline
    \diagbox{$(\omega^A, \omega^B)$}{$(\theta^A, \theta^B)$} & $(1,3)$ & $(1,4)$ & $(2,3)$ & $(2,4)$ \\
    %\cline{2-5}
    %& \ $3$ \ & \ $4$ \ & \ $3$ \ & \ $4$ \ \\
    \hline
    $(+1,+1)$ & $\mu$ & $\nu$ & $\mu$ & $\mu$ \\
    \hline
    $(-1,-1)$ & $\mu$ & $\nu$ & $\mu$ & $\mu$ \\
    \hline
    $(+1,-1)$ & $\nu$ & $\mu$ & $\nu$ & $\nu$ \\
    \hline
    $(-1,+1)$ & $\nu$ & $\mu$ & $\nu$ & $\nu$ \\
    \hline
\end{tabular}
\caption{CHSH table specifying the odds $p(\omega^A, \omega^B \mid \theta^A, \theta^B)$ for each setting pair $(\theta^A, \theta^B)$ and outcome $(\omega^A, \omega^B)$, where $\mu := (2- \sqrt{2})/8$ and $\nu := (2+ \sqrt{2})/8$.} \label{table:CHSH_odds}
\end{center}
\end{table}

The key observation is that for each fixed measurement pair $(s, t) \in \Theta$, the distribution $p(\,\cdot, \cdot \mid s, t)$ is the element of $\mathcal{P}(\Omega^2)$, that is,
\[
\sum_{(a, b) \in \Omega^2} p(a, b \mid s, t) = 2 \mu + 2 \nu = 1.
\]
Thus, at the level of each fixed measurement pair $(s, t)$, the situation admits a fully classical probabilistic description and no inconsistency appears.

The difficulty arises only when one asks whether these four distributions can be realized simultaneously within a single hidden-variable model, that is, whether there exists a probability space $\Lambda$ and random variables $X_1,\dots,X_4 : \Lambda \to \Omega$ whose pairwise marginals reproduce $p(\,\cdot, \cdot \mid s, t)$ for all $(s,t)\in\Theta$.
Indeed, the existence of such a single joint distribution would imply the CHSH inequality 
\[
|C(1,3) - C(1,4) + C(2,3) + C(2,4)| \leq 2,
\]
where, 
\[ C(s,t):=\sum_{(a,b) \in\Omega^2} a b \cdot p(a, b \mid s, t) \]
denotes the correlation (see Appendix \ref{app:CHSH_summary} for details). 
However, the parameters $(\mu, \nu)$ specified in Table \ref{table:CHSH_odds} do not satisfy CHSH inequality.
In fact, the CHSH table yields
\[
C\left(1, 3 \right) = -C(1,4)=C(2,3)=C(2,4) = -\frac{1}{\sqrt{2}},
\]
and consequently
\[
C \left( 1, 3 \right) - C \left( 1, 4 \right) + C \left( 2, 3 \right) + C \left( 2, 4 \right) = -2\sqrt{2},
\]
which violates the CHSH bound.
Equivalently, there is no single underlying probability space (i.e., no joint distribution) whose marginals reproduce the CHSH table for all $(s,t)\in\Theta$ simultaneously.

%===========================================================================================
\section{Game-theoretic analysis of CHSH experiments with and without loopholes}
%===========================================================================================

In this section we recast two major loopholes in CHSH-type experiments, the locality loophole and the measurement-dependence loophole, as {\it structural} constraints in a hidden-variable game.
We start from a baseline game in which Nature is explicitly allowed to exploit the locality loophole, in the sense that Nature may choose the hidden variable after obtaining partial information from one wing of the experiment.
We then identify what additional restrictions on the timing and availability of information are needed to exclude both loopholes, and show that imposing these restrictions leads naturally to the loopholes-closed game studied in the main theorem.

\subsection{A hidden-variable game with the locality loophole}

We first address the locality loophole.
Our aim is to model a situation in which, when Nature selects the hidden variable \(\lambda\), information about A's measurement setting \(s\) and outcome \(\omega^{s}\), and possibly also B's setting, may already be available.
In such a scenario Nature can correlate \(\lambda\) with the settings and/or outcomes, thereby simulating correlations that would otherwise be ruled out by locality-based arguments.
Operationally, this corresponds to experimental situations in which the two measurements are not space-like separated (e.g., they are not effectively simultaneous), so that setting or outcome information can be transmitted between the two wings before \(\lambda\) is fixed.

To formalize this situation into a game-theoretic framework, we consider the following game.
We define, for any $a, b \in \Omega$, $s \in \Theta^A$, and $t \in \Theta^B$,
\begin{align*}
p(a \mid s) &:= \sum_{b' \in \Omega} p(a, b' \mid s, 3) =  \sum_{b' \in \Omega} p(a, b' \mid s, 4) = \frac{1}{2}, \\
p(b \mid a, s, t) &:= \frac{p(a, b \mid s, t)}{p(a \mid s)} = 2 p(a, b \mid s,t).
\end{align*}

\medskip
\begin{itembox}[l]{\bf Hidden-variable game with the locality loophole}
\textbf{Players} : Scientist A, Scientist B, Nature A, and Nature B. \\
\textbf{Protocol} : $F_0^{A} = F_0^B =  1$. \\
\hspace*{0.5em}FOR $n \in \mathbb{Z}_{>0}$ : \\
\hspace*{1em}Scientist A announces $\{q_n^A(\,\cdot\, \mid s ) \in \mathcal{P}(\Omega)\}_{s \in \Theta^A}$ and $s_n \in \Theta^A$. \\
\hspace*{1em}Nature A announces $\omega_n^{s_n} \in \Omega$. \\
\hspace*{1em}Scientist B announces $\{ q_n^B(\,\cdot\, \mid \omega_n^{s_n}, s_n, t) \in \mathcal{P}(\Omega) \}_{t \in \Theta^B}$ and $t_n \in \Theta^B$. \\
\hspace*{1em}Nature B chooses $\lambda_n \in X_{s_n}^{-1} (\{\omega_n^{s_n}\})$ and announces $\omega_n^{t_n} := X_{t_n} (\lambda_n)$. 
\smallskip \\
\hspace*{1em}$\displaystyle F_n^{A} := F_{n-1}^{A} \cdot \frac{ q_n^A(\omega_n^{s_n} \mid s_n)}{ p(\omega_n^{s_n} \mid s_n)}$,
\smallskip \\
\hspace*{1em}$\displaystyle F_n^{B} := F_{n-1}^{B} \cdot \frac{ q_n^B(\omega_n^{t_n} \mid \omega_n^{s_n}, s_n, t_n) }{ p(\omega_n^{t_n} \mid \omega_n^{s_n}, s_n, t_n)}$. \\
\hspace*{0.5em}END FOR.
\end{itembox}
\medskip

This game is intended to model the locality loophole as follows.
Scientist~A first specifies a betting distribution on \(\Omega\) for each setting \(s\), and then chooses the setting \(s_n\).
After observing A's move, Nature~A produces the corresponding outcome \(\omega_n^{s_n}\).
Next, Scientist~B specifies a betting distribution that may depend on the pair \((s_n,\omega_n^{s_n})\) and then chooses the setting \(t_n\).
Finally, Nature~B is allowed to use the information available from both wings---in particular, the setting and outcome on A's side---to select a hidden variable \(\lambda_n\) consistent with A's realized outcome, and to generate B's outcome via \(\omega_n^{t_n}=X_{t_n}(\lambda_n)\).

The key feature is that Nature~B may choose \(\lambda_n\) {\it after} learning \((s_n,\omega_n^{s_n})\).
This explicitly captures the operational content of the locality loophole: information can propagate between the two wings before the hidden variable is effectively fixed, allowing correlations that would be excluded under strict locality (space-like separation).

The following theorem holds.

\begin{thm} \label{thm:CHSH_LLN_with_commu_loophole}
In the hidden-variable game with the locality loophole,
Scientists A and B have a joint strategy that ensures $\lim_{n \to \infty} F_n^A = \infty$ or $\lim_{n \to \infty} F_n^B = \infty$ unless,
for any $s \in \Theta^A$, $t \in \Theta^B$ and $a, b \in \Omega$,
\[
\lim_{n \to \infty} \frac{ \sum_{i=1}^n \delta_{(a, b, s, t)} (\omega_i^{s_i}, \omega_i^{t_i}, s_i , t_i )}{ \sum_{i=1}^n \delta_{(s,t)} (s_i, t_i)} = p(a , b \mid s, t).
\]
\end{thm}

\begin{proof}
See Section \ref{proof_1}. 
\end{proof}

This theorem implies that Scientists have a joint betting strategy that allows them to make their capital diverge if Nature's move, namely, the empirical frequencies of outcomes, fails to converge to the CHSH table.
In the sense of Shafer and Vovk \cite{itsonlyagame}, Theorem~\ref{thm:CHSH_LLN_with_commu_loophole} states that if Nature, following the hidden-variable model, is allowed to choose $\lambda$ after observing both the measurement setting and the outcome on A's side, then Scientists can force Nature to reproduce empirical correlations consistent with the CHSH table.
From this perspective, the locality loophole is precisely the structural feature that prevents the CHSH inequality from emerging as a game-theoretic law of large numbers.
Similarly, one may construct a game where Nature can exploit the measurement-dependence loophole.
In that case, an analogous theorem holds for the corresponding hidden-variable game.

\subsection{A hidden-variable game with loopholes closed}

We next consider the game inspired by the CHSH experiment with loopholes closed.
By closing the currently considered loopholes, is Nature, following a hidden variable model, forced to exhibit behavior consistent with the CHSH table?

To address this question, we examine structural requirements to close the locality loophole and measurement-dependence loophole.
One naive game-theoretic method for closing the locality loophole is to adopt a protocol where, {\it after} Nature declares $\lambda$, Scientists determine the measurement setting $(s, t)$.
Experimentally, this can be viewed as analogous to a Bell-type delayed-choice experiment \cite{Xiao_Song_Ma, Chaves}, that is, a scheme in which the choice of measurement directions is delayed until after the particle is generated and enters the observation apparatus.
Thus, in this formulation, Nature is no longer the last player to act.
This seems somewhat conceptually incompatible with the usual role of Nature in game-theoretic probability, where Nature is modeled as the last player to move.
Furthermore, this formulation imposes additional requirements to close the measurement-dependence loophole, since the measurement settings might have correlations with $\lambda$.
This is also unsatisfactory from a game-theoretic probabilistic perspective.

To address these conceptual tensions, this paper reinterprets the formulation as ``the locality loophole does not occur if Nature behaves {\it as if} it does not anticipate the experimental outcome'' and embeds this into the game.
Likewise, the measurement-dependence loophole is resolved by having Nature behave ``{\it as if} it does not anticipate the measurement direction at all.''
In other words, rather than postulating a probabilistic independence like constraint between the hidden variable and the settings, we introduce an additional capital process into betting game, in which any deviation from such ``independence-like'' behavior allows the Scientists' capital to diverge.
Thus, Nature must act as though it cannot exploit information about outcomes or settings when declaring $\lambda$ in order to prevent Scientists from becoming infinitely rich.
Formalizing this operational constraint leads us to the ``loopholes-closed game.''

Before discussing the loopholes-closed game, we introduce several definitions that allow us to formulate an analogue of measurement independence in a game-theoretic, frequency-based manner.
Suppose that we are given sequences $\{\lambda_k \in \Lambda\}_{k=1}^n$ and $\{(s_k, t_k) \in \Theta \}_{k=1}^n$.
For $(\omega^1, \omega^2, \omega^3, \omega^4) \in \Omega^4$ and $(s, t) \in \Theta$, define
\begin{align*}
  U_n( \omega^1, \omega^2, \omega^3, \omega^4; s,t) 
  :=& \left| \left\{k \middle|
  \begin{alignedat}{2}
    \; &1 \leq k \leq n,\, s_k = s,\, t_k = t,
    \\ &(X_1(\lambda_k), X_2(\lambda_k), X_3(\lambda_k), X_4(\lambda_k)) = (\omega^1,\omega^2, \omega^3, \omega^4)
  \end{alignedat}
  \right\} \right|, \\
  V_n( \omega^1, \omega^2, \omega^3, \omega^4)
  :=& \left| \left\{k \middle|
  \begin{alignedat}{2}
    \; &1 \leq k \leq n,\,
    \\ &(X_1(\lambda_k), X_2(\lambda_k), X_3(\lambda_k), X_4(\lambda_k)) = (\omega^1,\omega^2, \omega^3, \omega^4)
  \end{alignedat}
  \right\} \right| \\
  =& \sum_{(s, t) \in \Theta} U_n(\omega^1, \omega^2, \omega^3, \omega^4; s,t), \\
  W_n(s, t)
  :=& | \{k \mid 1\leq k \leq n,\, s_k = s,\, t_k = t\} | \\
  =& \sum_{(\omega^1, \omega^2, \omega^3, \omega^4) \in \Omega^4} U_n(\omega^1, \omega^2, \omega^3, \omega^4; s,t).
\end{align*}
The quantity $U_n(\omega^1, \omega^2, \omega^3, \omega^4; s,t)$ counts how often the hidden variable assignments
produce the outcome quadruple $(\omega^1, \omega^2, \omega^3, \omega^4)$ under the measurement setting $(s,t)$.
The quantities $V_n(\omega^1, \omega^2, \omega^3, \omega^4)$ and $W_n(s,t)$ are the corresponding marginal counts.
We then define empirical distributions on $\Omega^4 \times \Theta$, $\Omega^4$, and $\Theta$ by
\begin{align*}
\hat{P}_n (\omega^1, \omega^2, \omega^3, \omega^4; s,t) &:= \frac{U_n(\omega^1, \omega^2, \omega^3, \omega^4; s,t)}{n}, \\
\hat{Q}_n (\omega^1, \omega^2, \omega^3, \omega^4) &:= \frac{V_n(\omega^1, \omega^2, \omega^3, \omega^4)}{n}, \\
\hat{R}_n (s,t) &:= \frac{W_n(s,t)}{n}. 
\end{align*}
In the usual probabilistic formulation of the CHSH inequality, measurement-independence assumption asserts that the distributions on $\Lambda$ and on $\Theta$ are independent in the measure-theoretic sense.
Under this assumption, one expects at least that,
\[
\lim_{n \to \infty} \left( \hat{P}_n -\hat{Q}_n \cdot \hat{R}_n \right) = 0 \quad \text{almost surely}.
\]
This observation motivates the construction of an additional capital process that tests, in an operational and sequential manner, whether such independence holds in the empirical data.

Then the following game provides a structural expression of closing the two loopholes within our framework.
Since there is no need to distinguish between Scientists A and B in this game, these two are combined into a single player, Scientist AB. The same applies to Nature AB.

\medskip
\begin{itembox}[l]{\bf Hidden-variable game with loopholes closed}
\textbf{Players} : Scientist AB, Nature AB. \\
\textbf{Protocol} : $F_0^{AB} = 1$, $I_0^{AB} = 1$. \\
\hspace*{0.5em}FOR $n \in \mathbb{Z}_{>0}$ : \\
\hspace*{1em}Scientist AB  announces $\{(q_n^{AB}(\,\cdot, \cdot \,\mid s, t ) ) \in \mathcal{P}(\Omega^2)\}_{(s, t) \in \Theta }$ and $(s_n, t_n) \in \Theta$.\\
\hspace*{1em}Nature AB announces $\lambda_n \in \Lambda$ and $\omega_n^s := X_s(\lambda_n)$ for $s \in \{1,2,3,4\}$.
\smallskip\\
\hspace*{1em}$\displaystyle F_n^{AB} := F_{n-1}^{AB} \cdot \frac{ q_n^{AB}(\omega_n^{s_n}, \omega_n^{t_n} \mid s_n, t_n) }{ p(\omega_n^{s_n}, \omega_n^{t_n} \mid s_n, t_n)}$.
\smallskip \\
\hspace*{1em}$\displaystyle I_n^{AB} := I_{n-1}^{AB} \cdot \frac{\hat{P}_n(\omega_n^1, \omega_n^2, \omega_n^3, \omega_n^4; s_n,t_n)}{\hat{Q}_n(\omega_n^1, \omega_n^2, \omega_n^3, \omega_n^4) \cdot \hat{R}_n(s_n,t_n)} $. \\
\hspace*{0.5em}END FOR.
\end{itembox}
\medskip

The second capital process $I_n^{AB}$ is designed to test, in a sequential manner, whether the empirical joint distribution on $\Omega^4 \times \Theta$ factorizes asymptotically into the product of its marginals.
It is defined recursively via a likelihood-ratio-type update between $\hat{P}_n$ and $\hat{Q}_n \cdot \hat{R}_n$, so that any systematic dependence between Nature's hidden-variable declarations and the chosen measurement settings leads to the divergence of $I_n^{AB}$.

Structurally, the present game differs from the locality-loophole game in only two respects.
First, Scientist AB introduces the additional capital process $I_n^{AB}$ in order to operationally enforce a measurement-independence-like constraint.
Second, by changing the timing of the announcements (relative to the locality-loophole game), we reformulate the dependence between the hidden variable $\lambda_n$ and the measurement settings $(s_n, t_n)$ from a purely temporal dependence into a dependence that is empirically testable through frequencies.
Although the timing of moves differs from that in the locality-loophole game, this does not restrict Nature's ability to generate outcome pairs $(\omega_n^{s_n}, \omega_n^{t_n})$ as functions of $(s_n, t_n)$.
The essential difference is that the loophole-closed game makes such dependence detectable within the game, which ultimately leads to the following incompatibility result.

\begin{thm}\label{thm:CHSH_holds}
    In the hidden-variable game with loopholes closed, Scientist AB has a strategy $q : \Omega^\ast \rightarrow \mathcal{P}(\Omega^2)$, that ensures $\lim_{n \to \infty} F_n^{AB} = \infty$ unless, for any $(s,t) \in \Theta$ and $a, b \in \Omega$,
    \begin{equation} \label{eqn:simulate_CHSH_table}
    \lim_{n \to \infty} \frac{ \sum_{i=1}^n \delta_{(a, b, s, t)} (\omega_i^{s_i}, \omega_i^{t_i}, s_i , t_i )}{ \sum_{i=1}^n \delta_{(s,t)} (s_i, t_i)} = p(a , b \mid s, t).
    \end{equation}
    In addition, Scientist AB's second capital process $I_n^{AB}$ diverges, $\lim_{n \to \infty} I_n^{AB} = \infty$, unless, for any $(\omega^1, \omega^2, \omega^3, \omega^4) \in \Omega^4$ and $(s, t) \in \Theta$,
    \begin{equation} \label{eqn:simulate_independence}
    \lim_{n \to \infty} \left\{ \hat{P}_n(\omega^1, \omega^2, \omega^3, \omega^4 ; s,t) -\hat{Q}_n(\omega^1, \omega^2, \omega^3, \omega^4) \cdot \hat{R}_n(s,t) \right\} = 0. 
    \end{equation}
    Furthermore, Nature AB cannot satisfy \eqref{eqn:simulate_CHSH_table} and \eqref{eqn:simulate_independence} simultaneously.
    Thus in this game, either $\lim_{n \to \infty} F_n^{AB} = \infty $ or $\lim_{n \to \infty} I_n^{AB} = \infty$ must occur.
\end{thm}

\begin{proof}
See Section \ref{proof_2}. 
\end{proof}

Theorem~\ref{thm:CHSH_holds} states, in game-theoretic terms, that Scientist~AB possesses a {\it winning strategy} in the hidden-variable game with loopholes closed.
More precisely, Scientist~AB can choose a strategy \(q:\Omega^\ast\to\mathcal{P}(\Omega^2)\) such that Nature~AB cannot keep both capital processes \(F_n^{AB}\) and \(I_n^{AB}\) bounded.
Scientists run two capital processes in parallel.
\begin{itemize}
  \item \(F_n^{AB}\) tests whether the empirical conditional frequencies converge to the CHSH table.
  \item \(I_n^{AB}\) tests whether the empirical distribution on \(\Omega^4\times\Theta\) factorizes asymptotically into the product of its marginal distributions on \(\Omega^4\) and on \(\Theta\), providing a frequency-based analogue of measurement independence.
\end{itemize}

Note that \eqref{eqn:simulate_independence} is weaker than the probabilistic independence assumptions usually imposed in the context of the CHSH inequality.
%Theorem~\ref{thm:CHSH_holds} asserts that Nature (the LHV side) cannot satisfy these two convergence requirements simultaneously.
Even under this weaker requirement, Nature must either reproduce the CHSH table at the level of empirical conditional frequencies \eqref{eqn:simulate_CHSH_table} {\it or} exhibit setting-independence-like behavior, in the frequency sense \eqref{eqn:simulate_independence}, but it cannot achieve both simultaneously.
Thus, as \(n\to\infty\), at least one of \(F_n^{AB}\) or \(I_n^{AB}\) necessarily diverges.

At first sight this may seem at odds with the empirical situation: in actual CHSH experiments, Bell-type inequalities are violated in accordance with quantum-mechanical predictions.
The point, however, is that the theorem is not merely saying that ``a violation occurs.''
Rather, it clarifies what must {\it give way} within a local hidden-variable (LHV) picture once the relevant loopholes are controlled.
Physically, an LHV model can survive asymptotically only by abandoning one of the two: either it fails to reproduce the CHSH table, or it gives up measurement independence by correlating the setting choices with the hidden variable.

Operationally, the two capital processes function as sequential test statistics that diagnose compatibility with each requirement along the data stream.
Thus, in a loophole-closed setting, Scientists obtain an operational winning strategy: they inevitably detect a breakdown, either in reproducing the CHSH table at the level of empirical frequencies or in setting-independence-like behaviour, thereby excluding LHV explanations in a data-driven sense.

%==================================================================
\section{Concluding remarks}
%==================================================================
We revisited the CHSH inequality from an explicitly information-structural perspective, emphasizing that the validity of Bell-type constraints depends not only on algebraic relations among correlations but also on who can access which information, and when.
To make this timing and availability explicit, we worked in the framework of game-theoretic probability, in which statistical claims are expressed through the existence of betting strategies rather than by positing an underlying probability space.

Our main contribution was to formalize two representative loopholes, the locality loophole and the measurement-dependence (``freedom-of-choice'') loophole, as structural constraints in sequential games inspired by hidden-variable scenarios.
Within this framework, we designed a loopholes-closed game and constructed two capital processes.
One monitors whether empirical conditional frequencies converge to the CHSH table, while the other monitors whether, across trials, the association between Nature's internally fixed outcome assignments and the chosen measurement settings remains free of systematic dependence.
This second test provides an operational, frequency-based surrogate for the usual measurement-independence assumption.
The central theorem shows that Nature cannot keep both processes bounded.
Equivalently, at least one capital process must diverge, so that Nature inevitably loses in the sense of game-theoretic probability.
This provides an operational reading of loophole-free CHSH violations: when the structural conditions corresponding to locality and freedom of choice are suitably controlled, Scientists possess a winning strategy against any hidden-variable explanation attempting to satisfy both requirements simultaneously.

Several directions for future work appear natural.
First, in the present paper we tested independence only at the level of single-trial empirical frequencies, which yields a deliberately weak, operational notion of measurement independence.
It would be interesting to investigate how the CHSH constraint changes when one imposes stronger regularity conditions, for example by monitoring correlations across multiple trials or higher-order configuration frequencies.
Such extensions may correspond, in game-theoretic terms, to strengthening the independence requirement beyond simple factorization of single-trial frequencies.
Second, it would be interesting to reinterpret various variants of Bell inequalities within the framework of game-theoretic probability.
Such reinterpretations may reveal structural or operational aspects that are less transparent in standard probabilistic formulations.
Finally, the present framework suggests a systematic way to compare different notions of ``loophole closure'' by translating them into game constraints; clarifying the relations among such constraints may help sharpen the interpretation of experimental protocols.

%==================================================================
\section*{Acknowledgements}
%==================================================================

The present study was supported by JST ERATO Grant Number JPMJER2402 and JSPS KAKENHI Grant Number 23K25787.

%==============================================================================
% Appendices ;----.----;----.----;----.----;----.----;----.----;----.----;----.----
%==============================================================================
\appendix
\section*{Appendix}
\setcounter{equation}{0}
\setcounter{thm}{0}
\setcounter{footnote}{1}
\addcontentsline{toc}{section}{Appendix}
\renewcommand{\thesubsection}{\Alph{subsection}}
\renewcommand{\thethm}{\Alph{subsection}.\arabic{thm}} %This changes the style of an already defined environment "thm" 
\renewcommand{\theequation}{\Alph{subsection}.\arabic{equation}}

%------------------------------------------------------------------------------------------------
\subsection{Proofs of Theorems presented in Section 4}
%------------------------------------------------------------------------------------------------

%-------------------------------------------------------------------------------------------------
\subsubsection{Proof of Theorem~\ref{thm:CHSH_LLN_with_commu_loophole}} \label{proof_1}
%-------------------------------------------------------------------------------------------------

Using Theorem~\ref{SLLN:predictive-Shafer_and_Vovk}, Scientist A can force Nature A to realize the event
\[
\lim_{n \to \infty} \frac{ \sum_{i=1}^n \delta_{(a, s)} (\omega_i^{s_i}, s_i)}{ \sum_{i=1}^n \delta_s (s_i)} = p(a \mid s),
\]
for any $(a, s) \in \Omega \times \Theta^A$.
Similarly, for fixed $(a, s) \in \Omega \times \Theta^A$, again by Theorem~\ref{SLLN:predictive-Shafer_and_Vovk}, Scientist B can force Nature B to realize the event
\[
\lim_{n \to \infty} \frac{ \sum_{i=1}^n \delta_{(a, b, s, t)} (\omega_i^{s_i}, \omega_i^{t_i}, s_i, t_i)}{ \sum_{i=1}^n \delta_{(a, s, t)}(\omega_i^{s_i}, s_i, t_i)} = p(b \mid a,s,t),
\]
for any $(b, t) \in \Omega \times \Theta^B$.

We remark that the denominators are determined by the declarations of Scientists A and B.
Hence, without loss of generality, we may assume that they are eventually nonzero and diverge.
In particular, Scientists can ensure that $\sum_{i=1}^n \delta_{(a,s,t)} (\omega_i^{s_i}, s_i, t_i) \to \infty$.

Combining these results, it remains to show that
\begin{equation} \label{eqn:freq_cond_by_B}
\frac{ \sum_{i=1}^n \delta_{(a, s)} (\omega_i^{s_i}, s_i)}{ \sum_{i=1}^n \delta_s (s_i)} \cdot \frac{ \sum_{i=1}^n \delta_{(a, b, s, t)} (\omega_i^{s_i}, \omega_i^{t_i}, s_i, t_i)}{ \sum_{i=1}^n \delta_{(a, s, t)}(\omega_i^{s_i}, s_i, t_i)} 
- \frac{ \sum_{i=1}^n \delta_{(a, b, s, t)} (\omega_i^{s_i}, \omega_i^{t_i}, s_i , t_i )}{ \sum_{i=1}^n \delta_{(s,t)} (s_i, t_i)} =o(1).
\end{equation}
Since Scientist B announces $t_n \in \Theta^B$, Scientist B can adopt a strategy $(t_i)$ ensuring that
\[
\frac{ \sum_{i=1}^n \delta_{(a, s)} (\omega_i^{s_i}, s_i)}{ \sum_{i=1}^n \delta_s (s_i)} - \frac{ \sum_{i=1}^n \delta_{(a, s, t)} (\omega_i^{s_i}, s_i, t_i)}{ \sum_{i=1}^n \delta_{(s,t)} (s_i, t_i)}  = o(1).
\]
Thus, under this strategy, \eqref{eqn:freq_cond_by_B} holds in the limit.

%--------------------------------------------------------------------------------------------------
\subsubsection{Proofs of Theorem~\ref{thm:CHSH_holds}} \label{proof_2}
%--------------------------------------------------------------------------------------------------

We divide the proof into three parts.
First, as a direct consequence of Theorem~\ref{SLLN:predictive-Shafer_and_Vovk},
Scientist AB has a strategy $q : \Omega^\ast \rightarrow \mathcal{P}(\Omega)$, that ensures $\lim_{n \to \infty} F_n^{AB} = \infty$ unless, for all $(s, t) \in \Theta$ and $a, b \in \Omega$,
\[
\lim_{n \to \infty} \frac{ \sum_{i=1}^n \delta_{(a, b, s, t)} (\omega_i^{s_i}, \omega_i^{t_i}, s_i , t_i )}{ \sum_{i=1}^n \delta_{(s,t)} (s_i, t_i)} = p(a , b \mid s, t).
\]
We note that as in Theorem~\ref{thm:CHSH_LLN_with_commu_loophole}, this step is merely an application of Theorem~\ref{SLLN:predictive-Shafer_and_Vovk}, which forces Nature to conform to the CHSH table at the level of empirical frequencies.
The novelty of Theorem~\ref{thm:CHSH_holds} lies not in this enforcement itself, but in the fact that its consequence is incompatible with the boundedness of the capital process $I_n^{AB}$.

Second, we evaluate the capital process $I_n^{AB}$.
To simplify the symbols, we write $\tau = (\omega^1, \omega^2, \omega^3, \omega^4) \in \Omega^4$ and $u = (s, t) \in \Theta$.
Then by the recursion formula of $I_n^{AB}$,
\[
I_n^{AB} = \frac{ n! \cdot \prod_{\tau \in \Omega^4} \prod_{u \in \Theta} U_n(\tau ; u)! }{
\prod_{\tau \in \Omega^4} V_n(\tau)! \cdot \prod_{u \in \Theta} W_n(u)! }.
\]
Then, taking the logarithm, and applying Stirling's formula in the second equation,
\begin{align*}
\frac{\log I_n^{AB}}{n}
=& \frac{1}{n} \left\{ \log n! + \sum_{\tau \in \Omega^4} \sum_{u \in \Theta} \log U_n(\tau; u)! - \sum_{\tau \in \Omega^4} \log V_n(\tau)! -\sum_{u \in \Theta} \log W_n(u)!  \right\} \\
=& \log n + \sum_{\tau, u} \frac{U_n(\tau; u)}{n} \log U_n(\tau;u) \\
& - \sum_{\tau} \frac{V_n(\tau)}{n} \log V_n(\tau) - \sum_{u} \frac{W_n(u)}{n} \log W_n(u) + O\left(\frac{\log n}{n}\right) \\
=& \sum_{\tau, u} \frac{U_n(\tau; u)}{n} \left( \log \frac{U_n(\tau;u)}{n} - \log \frac{V_n(\tau)}{n} - \log \frac{W_n(u)}{n} \right) + O\left(\frac{\log n}{n}\right) \\
=& D\left( \hat{P}_n \left\|  \hat{Q}_n \cdot \hat{R}_n  \right.\right) + O\left(\frac{\log n}{n}\right),
\end{align*}
where $D(\,\cdot\, \| \,\cdot\,)$ denotes the Kullback-Leibler divergence between probability distributions on $\Omega^4 \times \Theta$.
Hence, we get that if
\[
 \hat{P}_n - \hat{Q}_n \cdot \hat{R}_n \nrightarrow 0,
\]
then $\limsup I_n^{AB} = \infty$.

Finally, we show that Nature AB cannot satisfy \eqref{eqn:simulate_CHSH_table} and \eqref{eqn:simulate_independence} simultaneously.
Indeed, under the condition \eqref{eqn:simulate_independence}, the asymptotic behavior of Nature AB necessarily satisfies the CHSH inequality, which contradicts \eqref{eqn:simulate_CHSH_table}, since the CHSH table violates the CHSH inequality.
We now make this implication explicit by deriving an asymptotic version of the CHSH inequality from \eqref{eqn:simulate_independence}.

The condition \eqref{eqn:simulate_independence} can be written by
\[
 \hat{P}_n - \hat{Q}_n \cdot \hat{R}_n = o(1).
\]
Hence, by adopting a strategy satisfying $\liminf \hat{R}_n >0$,
the asymptotic relation \eqref{eqn:simulate_independence} is equivalent to
\begin{align*}
 &\frac{U_n(\omega^1, \omega^2, \omega^3, \omega^4; s,t)}{W_n(s,t)} - \frac{V_n(\omega^1, \omega^2, \omega^3, \omega^4)}{n} \\
 &\ \ =
  \frac{\hat{P}_n(\omega^1, \omega^2, \omega^3, \omega^4;s,t)}{\hat{R}_n(s,t)} -\hat{Q_n}(\omega^1, \omega^2, \omega^3, \omega^4)
 = o(1).
 \end{align*}
Fixing $(s, t)$ and $(\omega^s, \omega^t)$, and summing over the remaining indices, we obtain
\begin{equation} \label{eqn:another_indep}
 \frac{\widetilde{U_n}( \omega^s, \omega^t ;s,t )}{W_n(s,t)} - \frac{\widetilde{V_n}(\omega^s, \omega^t)}{n} = o(1),
\end{equation}
where
\begin{align*}
    \widetilde{U_n} ( \omega^s, \omega^t ; s, t ) &:= | \{k \mid 1\leq k \leq n,\, X_s(\lambda_k) = \omega^s,\, X_t(\lambda_k) = \omega^t,\, s_k = s,\, t_k = t \} |, \\
    \widetilde{V_n} ( \omega^s, \omega^t) &:= | \{k \mid 1\leq k \leq n,\, X_s(\lambda_k) = \omega^s,\, X_t(\lambda_k) = \omega^t\} |.
\end{align*}

By using \eqref{eqn:another_indep}, we prove that the behavior of Nature AB asymptotically satisfies the CHSH inequality.
The following computation parallels the standard derivation of the CHSH inequality, with empirical frequencies replacing probabilities.
Define the empirical correlation function by
\[
C_n(s, t) := \sum_{(\omega^s, \omega^t) \in \Omega^2} \omega^s \omega^t \cdot \frac{\widetilde{U_n}(\omega^s, \omega^t ; s, t )}{W_n(s,t)}.
\]
Then
\begin{align*}
   &C_n(1, 3) - C_n(1,4) \\
   & = \sum_{(\omega^1, \omega^3) \in \Omega^2} \omega^1 \omega^3 \cdot \frac{\widetilde{U_n}(\omega^1, \omega^3 ; 1, 3)}{W_n(1, 3)} - \sum_{(\omega^1, \omega^4) \in \Omega^2} \omega^1 \omega^4 \cdot \frac{\widetilde{U_n}(\omega^1, \omega^4 ; 1, 4)}{W_n(1, 4)} \\
   & = \sum_{(\omega^1, \omega^3) \in \Omega^2} \omega^1 \omega^3 \cdot \frac{\widetilde{V_n}(\omega^1, \omega^3)}{n} - \sum_{(\omega^1, \omega^4) \in \Omega^2} \omega^1 \omega^4 \cdot \frac{\widetilde{V_n}(\omega^1, \omega^4 )}{n} + o(1) \\
   & = \sum_{(\omega^1,\omega^2, \omega^3, \omega^4) \in \Omega^4} \omega^1 \omega^3 \cdot \frac{V_n(\omega^1, \omega^2, \omega^3, \omega^4)}{n} - \sum_{(\omega^1,\omega^2, \omega^3, \omega^4) \in \Omega^4} \omega^1 \omega^4 \cdot \frac{V_n(\omega^1,\omega^2,\omega^3, \omega^4)}{n} + o(1) \\
   &= \sum_{(\omega^1,\omega^2, \omega^3, \omega^4) \in \Omega^4} \omega^1 \omega^3 (1 \pm \omega^2 \omega^4) \cdot \frac{V_n(\omega^1, \omega^2, \omega^3, \omega^4)}{n} \\
   &\ \ - \sum_{(\omega^1,\omega^2, \omega^3, \omega^4) \in \Omega^4} \omega^1 \omega^4 (1 \pm \omega^2 \omega^3) \cdot \frac{V_n(\omega^1,\omega^2,\omega^3, \omega^4)}{n} + o(1),
\end{align*}
where the same choice of sign is taken throughout.
Hence
\begin{align*}
    &| C_n(1, 3) - C_n(1,4) | \\
    &\leq \sum_{(\omega^1,\omega^2, \omega^3, \omega^4) \in \Omega^4} |\omega^1 \omega^3 | (1 \pm \omega^2 \omega^4) \cdot \frac{V_n( \omega^1, \omega^2, \omega^3, \omega^4)}{n} \\
    &\ \ + \sum_{(\omega^1,\omega^2, \omega^3, \omega^4) \in \Omega^4} |\omega^1 \omega^4 | (1 \pm \omega^2 \omega^3) \cdot \frac{V_n( \omega^1, \omega^2, \omega^3, \omega^4)}{n} + o(1) \\
    &= \sum_{(\omega^1,\omega^2, \omega^3, \omega^4) \in \Omega^4} (1 \pm \omega^2 \omega^4) \cdot \frac{V_n( \omega^1, \omega^2, \omega^3, \omega^4)}{n} \\
    &\ \ + \sum_{(\omega^1,\omega^2, \omega^3, \omega^4) \in \Omega^4} (1 \pm \omega^2 \omega^3) \cdot \frac{V_n( \omega^1, \omega^2, \omega^3, \omega^4)}{n} + o(1) \\
    &= \sum_{(\omega^2, \omega^4) \in \Omega^2} (1 \pm \omega^2 \omega^4) \cdot \frac{\widetilde{V_n}(\omega^2, \omega^4)}{n}
    + \sum_{(\omega^2, \omega^3) \in \Omega^2} (1 \pm \omega^2 \omega^3) \cdot \frac{\widetilde{V_n}(\omega^2,\omega^3)}{n} + o(1)\\
    &= \sum_{(\omega^2, \omega^4) \in \Omega^2} (1 \pm \omega^2 \omega^4) \cdot \frac{\widetilde{U_n}(\omega^2, \omega^4 ; 2,4)}{W_n(2,4)}
    + \sum_{(\omega^2, \omega^3) \in \Omega^2} (1 \pm \omega^2 \omega^3) \cdot \frac{\widetilde{U_n}(\omega^2,\omega^3 ;2,3)}{W_n(2,3)} + o(1) \\
    & = 2 \pm \left\{ C_n(2,4) + C_n(2,3) \right\} + o(1).
\end{align*}
Thus we conclude that
\begin{align*}
    |C_n(1,3) - C_n(1,4) + C_n(2,3) + C_n(2,4)| \leq 2 + o(1).
\end{align*}
This completes the proof.

%--------------------------------------------------------------------------------------------------
\subsection{CHSH inequality and its assumption} \label{app:CHSH_summary}
%--------------------------------------------------------------------------------------------------
This section summarizes the CHSH inequality and related topics for the reader's convenience and clarifies terminology usage.
First, we give the mathematical formulation and proof of the CHSH inequality.
Next, we explain how this assumption is interpreted in physical implementations.
We finally recall how quantum mechanics gives rise to correlations that violate the inequality.

We first introduce the CHSH inequality in a measure-theoretic setting.
Let $X_s$, for $s \in \{1,2,3,4\}$, be a real-valued random variables on a probability space $(\Lambda, \mathcal{F}, P)$.

For $s, t \in \{1,2,3,4\}$, define the correlations
\[
C(s,t) := \mathbb{E}[X_s \cdot X_t] = \int_\Lambda X_s(\lambda) X_t(\lambda) dP(\lambda),
\]
then the following inequality holds.

\begin{thm}[CHSH inequality]
Assume that, for any $s \in \{1,2,3,4\}$, $|X_s| \leq 1$ almost surely.
Then
\[
|S| \leq 2,
\]
where
\[
S := C(1,3) - C(1,4) + C(2,3) + C(2,4).
\]
\end{thm}

\begin{proof}
By the definition,
\begin{align*}
C(1,3) - C(1,4) &= \mathbb{E}[X_1 X_3] - \mathbb{E}[X_1 X_4] \\
& = \mathbb{E}[X_1 X_3 \cdot (1 \pm X_2 X_4)] - \mathbb{E}[X_1 X_4 \cdot (1 \pm X_2 X_3)],
\end{align*}
where the same choice of sign is taken throughout.
Due to the assumptions that $1 \pm X_2 X_4, 1\pm X_2 X_3 \geq 0$ a.s. $P$,
\begin{align*}
    |C(1,3) - C(1,4)|
    &\leq \mathbb{E}[ |X_1 X_3| \cdot (1 \pm X_2 X_4)] + \mathbb{E}[|X_1 X_4| \cdot (1 \pm X_2 X_3)] \\
    & \leq 2 \pm \left\{ C(2, 4) + C(2,3) \right\}.
\end{align*}
Hence we get
\begin{align*}
C(1,3) - C(1,4) &\leq 2 - \left\{ C(2, 4) + C(2,3) \right\}, \text{ and } \\
C(1,3) - C(1,4) & \geq -2 - \left\{ C(2, 4) + C(2,3) \right\}.
\end{align*}
Thus we conclude that
\[
|S| \leq 2.
\]
\end{proof}

In the measure-theoretic formulation above, the fact that the random variables $X_1, \dots , X_4$ are defined on a single probability space is simply a structural assumption on the model.
When these variables are interpreted as outcomes of measurements labeled by $s \in \{1,2,3,4\}$, this structural assumption is reinterpreted as the requirement that all four outcome values are simultaneously defined when $\lambda$ is fixed, including those for settings that are not actually chosen in a given run of the experiment.
These hypothetically assigned values are often referred to as counterfactual outcomes, in the sense that they correspond to outcomes of measurements that could have been performed but were not.
In this sense, a hidden-variable model assumes that every experimental run carries predetermined values for all possible measurement settings, even though only one pair of them is observed.

On the other hand, within quantum mechanics, one can design experiments that violate the CHSH inequality.
Let $\mathcal{H}_A = \mathcal{H}_B = \mathbb{C}^2$ and let $\rho$ be a maximally entangled pure state on $\mathcal{H}_A \otimes \mathcal{H}_B$.
For a standard choice of four local dichotomic observables corresponding to measurement settings on $\mathcal H_A$ and $\mathcal H_B$ that maximize the CHSH violation, one obtains
\[
C(1, 3)- C(1, 4) + C(2 , 3) + C(2 , 4) = -2\sqrt{2}.
\]
The resulting joint outcome probabilities are consistent with the CHSH table used in the main text.
For further details on Bell's theorem and the CHSH inequality in the quantum-mechanical setting, see \cite{Nielsen_Chuang}.

%==================================================================

\end{document}